 \documentclass[9pt, pdfmode]{article}

\usepackage{amsmath,amsfonts,amsthm,amssymb,amscd,cancel,color}

\usepackage{verbatim}
\usepackage{graphicx}
\renewcommand{\Re}{\mathop{\rm Re}\nolimits}
\renewcommand{\Im}{\mathop{\rm Im}\nolimits}

\theoremstyle{plain}
\newtheorem{theorem}{Theorem}[section]
\newtheorem{lemma}[theorem]{Lemma}
\newtheorem{proposition}[theorem]{Proposition}
\newtheorem{corollary}[theorem]{Corollary} 
\theoremstyle{definition}
 
\theoremstyle{remark}
\newtheorem{remark}[theorem]{Remark}

\newcommand{\R}{{\mathbb R}}

\newcommand{\Z}{{\mathbb Z}}

\newcommand{\N}{{\mathbb N}}

\def\im{{\rm i}}

\newcommand{\C}{\mathbb{C}}
\newcommand{\T}{\mathbb{T}}

\def\({\left(}
\def\){\right)} 
\def\<{\left\langle}
\def\>{\right\rangle}

\numberwithin{equation}{section}

\newcommand{\ur}{u_{\uparrow}}
\newcommand{\ul}{u_{\downarrow}}

\newcommand{\ai}{\alpha_{\infty}}

\newcommand{\ri}{\rho_{\infty}}

\newcommand{\jve}{J_{\mathrm{VE}}}
\newcommand{\jev}{J_{\mathrm{EV}}}

\begin{document}

\title{Absence of singular continuous spectra and embedded eigenvalues for one dimensional quantum walks with general long-range coins}

\author {Masaya Maeda, Akito Suzuki, Kazuyuki Wada}

\maketitle


\begin{abstract}
This paper is a continuation of the paper \cite{W} by the third author, which studied quantum walks with special long-range perturbations of the coin operator.
In this paper, we consider general long-range perturbations of the coin operator and prove the non-existence of a  singular continuous spectrum and embedded eigenvalues.
The proof relies on the construction of generalized eigenfunctions (Jost solutions) which was studied in the short-range case in \cite{MSSSSdis}.
\end{abstract}

%
\section{Introduction}
Quantum Walks (QWs), which are usually considered to be the quantum counterpart of classical random walks \cite{ADZ, ANVW, G, M} are now attracting diverse interest due to the connections with various fields in mathematics and physics. From a numerical analysis point of view, QWs are splitting method of Dirac equations, for splitting method, see e.g.\ \cite{HLW10GNI}. Indeed, the first model of QWs which appears in Feynman's textbook is obtained by discretization of the generator of Dirac equation in 1D \cite{FH}. It seems that the connection between Dirac equations and QWs is getting more and more important \cite{ANF, BES, LKN, MB, MBD, MDB, MaSu, Sh, St}. 
More directly, it has begun to be realized that many tools for the study of dispersive equations, such as the Dirac equations and the more well-studied Schr\"odinger equations, are transferable to the study of QWs.
More specifically, the spectral and scattering theory for QWs were studied for short-range perturbations \cite{ABC, ABJ, CGMV, FFSI, FFSII, HM, MoSe, MSSSSdis, RSTI, RSTII, S, W}. For other directions, such as quasi-periodic perturbations and random perturbations, see \cite{FO17JFA, K, YKE}.

In this paper, we study the spectral 
theory for long-range perturbations.
In particular, we show the absence of singular continuous spectrum and embedded eigenvalues, which are fundamental problems of spectral theory. 
To establish such result, a standard strategy is to apply the commutator theory for unitary operators \cite{ABC, ABJ, RSTI, W}.
However, in this paper, we construct Jost solutions with modified phases, which also called modified plane waves.
This stratedy, which is only applicable for one dimensional problems, gives the strongest result because it is based on ordinary differential (actually difference) equations.
Our method also seems to be applicable to one dimensional Dirac equation with arbitrary long-range potential, which are no results up to the authors knowledge.
  
We note that the long-range perturbation was first considered by the third author for a special class of perturbations \cite{W}, which generates no long-range phase modification.

\subsection{Set up}
For a $\C^2$-valued map $u:\Z\to \C^2$, we write
\begin{align*}
u(x)=\begin{pmatrix} \ur(x) \\ \ul(x) \end{pmatrix}, \quad x\in\Z.
\end{align*}
For Banach spaces $X$ and $Y$, we denote the Banach space of all bounded linear operators from $X$ to $Y$ by $\mathcal L(X,Y)$.
We also set $\mathcal L(X):=\mathcal L(X,X)$.
For $u:\Z\to \C$, we define $\tau u (x):=u(x-1)$, $x\in\mathbb{Z}$.

Let $\mathcal H:=l^2(\Z,\C^2)$.
We define the shift operator $S\in \mathcal L(\mathcal H)$ by
\begin{align}\label{def:shift}
\(S \begin{pmatrix} \ur \\ \ul \end{pmatrix}\)(x):=\(\begin{pmatrix} \tau & 0 \\ 0 & \tau ^{-1} \end{pmatrix} \begin{pmatrix} \ur \\ \ul \end{pmatrix}\)(x)=\begin{pmatrix} \ur(x-1) \\ \ul(x+1) \end{pmatrix},\hspace{5mm}x\in\mathbb{Z}.
\end{align}
Let $\alpha$ and $\beta$ be $\C$-valued functions on $\Z$ and $\theta$ be $\R$-valued function on $\Z$. For $\alpha$ and $\beta$, we impose the following condition:
$$|\alpha(x)|^2+|\beta(x)|^2=1\quad x\in\Z.$$

We define the coin operator $C=C_{\alpha,\beta,\theta}\in \mathcal L(\mathcal H)$ by
\begin{align}\label{def:coin:pre}
Cu(x):=C_{\alpha(x),\beta(x),\theta(x)}u(x),\quad C_{\alpha(x),\beta(x),\theta(x)}:=e^{\im \theta(x)} \begin{pmatrix} \beta(x) & \overline{\alpha(x)} \\ -\alpha(x) & \overline{\beta(x)} \end{pmatrix} \in U(2),
\end{align}
where $U(2)$ is the set of all $2\times 2$ unitary matrices. We define the time evolution operator of a quantum walk (QW) as
\begin{align}\label{def:U}
Uu:=SCu.
\end{align}
We are interested in the situation that the coin operator converges to a fixed unitary matrix as $|x|\to \infty$. 
In particular, we consider the following ``long-range" situation:
\begin{align}\label{ass:long:pre}
\begin{cases}
\alpha(x)\in B_{\C}(0, 1):=\{z\in\C|\ |z|<1\},
\\
\sum_{x\in\Z}\(|\alpha(x+1)-\alpha(x)|+|\theta(x+1)-\theta(x)|\)<\infty,
\\
|\alpha(x)-\ai|+|\theta(x)|\to 0,\ |x|\to \infty,\quad \ai\in B_\C(0,1)\setminus\{0\}.
\end{cases}
\end{align}

\begin{remark}
For \eqref{ass:long:pre}, the first assumption means that diagonal entries of $C_{\alpha(x), \beta(x), \theta(x)}$ do not vanish for all $x\in \Z$. The second assumption implies that there exist $\alpha_\pm\in \overline{B_\C(0,1)}$ and $\theta_\pm\in\R$ such that \ $|\alpha(x)-\alpha_\pm|+|\theta(x)-\theta_\pm|\to 0$ as $x\to \pm\infty$. The last assumption is added to ensure $\alpha_+=\alpha_- \in B_\C(0,1)\setminus\{0\}$ and $\theta_+=\theta_-=0$. We remark that there is no loss of generality assuming $\theta_+=0$.
\end{remark}

It is known that all $2\times 2$ unitary matrix can be represented like as $C_{\alpha(x), \beta(x), \theta(x)}$ in \eqref{def:coin:pre}. By the following proposition, we can simplify the form of the coin operator:
\begin{proposition}\label{gauge}
Suppose that $\alpha(x)\in B_{C}(0, 1)$ for all $x\in\Z$. Then, there exists a unitary operator $G\in\mathcal{L}(\mathcal{H})$ such that
\begin{align}
GUG^{-1}=SC_{\alpha'},\quad C_{\alpha'(x)}:=\begin{pmatrix}\sqrt{1-|\alpha'(x)|^2} & \overline{\alpha'(x)} \\ -\alpha'(x) & \sqrt{1-|\alpha'(x)|^2}\end{pmatrix},\quad \alpha'(x)\in B_{\C}(0, 1).
\end{align}
Moreover, under \eqref{ass:long:pre}, it follows that
$$ \displaystyle\sum_{x\in\Z}|\alpha'(x+1)-\alpha'(x)|<\infty,\quad\alpha'(x)\rightarrow \alpha'_\pm\quad \text{as }x\rightarrow\pm\infty,\quad |\alpha'_{+}|=|\alpha'_{-}|=|\alpha_{\infty}|.$$
\end{proposition}
\begin{remark}
In Proposition \ref{gauge}, there is a poosibility of $\alpha'_{+}\neq \alpha'_{-}$ even if we assume $\alpha_{+}=\alpha_{-}$.
\end{remark}

We give proof of Proposition \ref{gauge} in Appendix A. From Proposition \ref{gauge}, it suffices to consider the following situation.

We rewrite the coin operator $C$ as
\begin{align}\label{def:coin}
Cu(x):=C_{\alpha(x)}u(x),\quad C_{\alpha(x)}:=\begin{pmatrix} \rho(x) & \overline{\alpha(x)} \\ -\alpha(x) & \rho(x) \end{pmatrix} \in U(2), \quad \rho(x):=\sqrt{1-|\alpha(x)|^2}.
\end{align}


Moreover, we rewrite the condition \eqref{ass:long:pre} as
\begin{align}\label{ass:long}
\begin{cases}
\alpha(x)\in B_{\C}(0, 1),\quad 
\\
\sum_{x\in\Z}|\alpha(x+1)-\alpha(x)|<\infty,\quad \alpha(x)\rightarrow\alpha_{\pm}\quad (\text{as }x\rightarrow\pm\infty), 
\\
|\alpha_{+}|=|\alpha_{-}|,\quad 0<|\alpha_{\pm}|<1.
\end{cases}
\end{align}
In what follows, we consider the time evolution operator $U:=SC=SC_{\alpha}$.

\begin{remark}
\begin{enumerate}
\item The assumption $0<|\alpha_{\pm}|<1$ (also $0\neq \alpha_{\infty}\in B_{\C}(0, 1)$) is needed to construct a modified frequency $\zeta$ which will be introduced in (\ref{pointfreq}). If $|\alpha_{\pm}|=1$, The essential spectrum of $U$ coincides with $\{\im, -\im\}$. In this case, it seems that a quantum walker does not scatter. If $\alpha_{\pm}=0$, the spectrum of $U$ coincides with the unit circle on $\C$. In this case, it seems to be difficult to construct $\zeta_{\alpha}$ since the dispersion relation \eqref{disprel} is quite different from the case of $0<|\alpha_{\pm}|<1$. We left it for future study. 
\item The form of the coin operator \eqref{def:coin} looks like strictly restricted. However, we would like to emphasize that the coin operator \eqref{def:coin:pre} under \eqref{ass:long:pre} is contained in \eqref{def:coin} under \eqref{ass:long}.
\end{enumerate}
\end{remark}



We set $U_{\mathrm{r}}:=SC_{\mathrm{r}}$ and $U_{\mathrm{l}}:=SC_{\mathrm{l}}$, where $C_{\mathrm{r}}$ and $C_{\mathrm{l}}$ are given by
\begin{align}\label{inf:coin}
C_{\mathrm{r}}(x):=\begin{pmatrix}\rho_{\infty} & \overline{\alpha_{+}} \\ -\alpha_{+} & \rho_{\infty}\end{pmatrix}, \quad C_{\mathrm{l}}(x):=\begin{pmatrix}\rho_{\infty} & \overline{\alpha_{-}} \\ -\alpha_{-} & \rho_{\infty}\end{pmatrix},\quad x\in\Z,
\end{align}
where $\rho_{\infty}:=\sqrt{1-|\alpha_{\pm}|^2}$. By the discrete Fourier transformation, it is seen that 
$$ \sigma(U_{\mathrm{r}})=\sigma(U_{\mathrm{l}})=\{e^{\im\lambda}|\ |\cos\lambda|\le \rho_{\infty}\}.$$
From Theorem 2.2 of \cite{RSTI}, it follows that
\begin{align}\label{ess}
\sigma_{\mathrm{ess}}(U)=\sigma(U_{\mathrm{r}})\cup \sigma(U_{\mathrm{l}})=\{e^{\im\lambda}|\ |\cos\lambda|\le \ri\},
\end{align}
where $\sigma_{\mathrm{ess}}(U)$ is the essential spectrum of $U$.

\subsection{Main Results}

Our main result is the following.
\begin{theorem}\label{main:spec}
Assume \eqref{ass:long}.
Let $\sigma_{\mathrm{sc}}(U)$ and $\sigma_{\mathrm{e}}(U)$ be the set of a singular continuous spectrum and eigenvalues of $U$, respectively.
Then, we have
\begin{align}\label{no:sing}
\sigma_{\mathrm{sc}}(U)=\emptyset,
\end{align}
and
\begin{align}\label{no:emb}
\sigma_{\mathrm{e}}(U) \cap \{e^{\im \lambda}\ |\ |\cos \lambda|< \rho_\infty\}=\emptyset.
\end{align}
\end{theorem}

\begin{remark}
	From \eqref{no:sing} and \eqref{no:emb}, the interior of essential spectrum of $U$ corresponds to the absolutely continuous spectrum.
\end{remark}

To prove Theorem \ref{main:spec}, we construct modified plane waves of $U$ and apply the limiting absorption principle for unitary operators. 
For short-range cases, such strategy was adapted in \cite{MSSSSdis} (see also \cite{HM}).
In \cite{MSSSSdis}, the idea to construct the plane wave was to rewrite the generalized eigenvalue problem as a difference equation using the transfer matrix (TM), see Proposition \ref{prop:equiv}, below.
Then, considering the TM as the perturbation of the TM of the constant coin case, they were able to construct the solution with the desired asymptotics because the perturbation was $l^1$.

Apparently, for the construction of modified plane waves, considering the TM as a perturbation of the TM with a constant coin is not sufficient.
Indeed, to construct the modified phase $Z$ (see \eqref{eq:modphase}), we need to diagonalize the transfer matrix at each point sufficiently near infinity and subtract the remainder to make the perturbation be in $l^1$.
This simple but new strategy requires detailed analysis for the  dispersion relation of QWs and more than half of the proof of this paper is devoted to such analysis.
We note that this detailed analysis of dispersion relation would be useful for studying other properties of QWs such as scattering and inverse scattering theory. In particular, Jost functions constructed in this paper will be applied to establish the weak limit theorem for quantum walks with general long-range coins which contain a result \cite{W}. 
After obtaining the estimates of the dispersion relation, the construction of modified plane waves will be more or less similar to \cite{MSSSSdis}.

The rest of this paper is organized as follows. In Sections \ref{subsec:tm} and \ref{subsec:constant}, we introduce the TM from the generalized eigenvalue problem and study the dispersion relation. In Section \ref{subsec:modwave}, we construct modified plane waves which are solutions to the generalized eigenvalue problem. In Section \ref{sec:prmain}, we prove the main result by using modified plane waves and the limiting absorption principle. In Appendix A, we give proof of Proposition \ref{gauge}.

\section{Generalized eigenvalue problem}\label{sec:genev}
In this section, we study the following (generalized) eigenvalue problem:
\begin{align}\label{ev:prob}
U u = e^{\im \lambda}u,\quad \lambda\in \C/2\pi\Z,
\end{align}
where $u:\Z\to \C^2$ is not necessarily in $\mathcal H=l^2(\Z,\C^2)$.

\subsection{Transfer matrix}\label{subsec:tm}

We recall several results from \cite{MSSSSdis}.
We set a unitary operator $\jve$ by
\begin{align}\label{def:jev}
\jve u(x):=\begin{pmatrix} \ul(x-1) \\ \ur(x) \end{pmatrix}.
\end{align}
We denote the inverse of $\jve$ by 
\begin{align}\label{def:jve}
\jev u(x):=\jve^{-1} u(x)=\begin{pmatrix}\ul(x) \\ \ur(x+1)\end{pmatrix}.
\end{align}
For $\lambda\in  \C/2\pi\Z$ and $x\in\Z$, we set
\begin{align}\label{def:tla}
T_\lambda(x):=\rho(x)^{-1}\begin{pmatrix}  e^{\im \lambda} & \alpha(x) \\ \overline{\alpha(x)} & e^{-\im \lambda} \end{pmatrix}.
\end{align}

The generalized eigenvalue problem \eqref{ev:prob} can be rewritten using transfer matrix $T_{\lambda}$.

\begin{proposition}[Proposition 3.2 of \cite{MSSSSdis}]\label{prop:equiv}
For $u:\Z\to \C^2$, the generalized eigenvalue problem \eqref{ev:prob} is equivalent to
\begin{align}\label{eq:equiv}
\(\jve u\)(x+1)=T_\lambda(x)\(\jve u\)(x),\ \forall x\in\Z.
\end{align}

\end{proposition}



As for ordinary differential equations, the Wronskian is invariant. The same property also follows in the present case.

\begin{proposition}[Proposition 3.3 of \cite{MSSSSdis}]\label{prop:Wronsky}
Let $v_1=\jve u_1$ and $\ v_2=\jve u_2$ satisfy \eqref{eq:equiv}.
Then, $\det (v_1(x)\ v_2(x))$ is independent of $x\in\Z$.
\end{proposition}


Corollary \ref{cor:Wron} below will be used to exclude the existence of embedded eigenvalues.

\begin{corollary}[Corollary 3.4 of \cite{MSSSSdis}]\label{cor:Wron}
Suppose that $u_1$ and $u_2$ satisfy \eqref{ev:prob}, $u_1$ is bounded on $\Z_{\geq 0}:=\{x\in\Z\ |\ x\geq 0\}$ and $u_2\in \mathcal H$.
Then, $u_1$ and $u_2$ are linearly dependent.
\end{corollary}


The kernel of the resolvent of $U$ can be written down using generalized eigenfunctions.
For $v\in \C^2$ (column vector), we set $v^\top$ (row vector) to be the transposition of $v$.

\begin{proposition}[Proposition 3.6 of \cite{MSSSSdis}]\label{prop:kernel}
Let $v_1=\jve u_1$ and $v_2=\jve u_2$ satisfy \eqref{eq:equiv} with $W_\lambda=\det(v_1\ v_2)\neq 0$.
We set
\begin{align*}
K_\lambda(x,y):=e^{-\im \lambda} W_\lambda^{-1} \(v_2(x) v_1(y)^\top \begin{pmatrix} 0 & 1_{<y}(x)\\ 1_{\leq y}(x) & 0 \end{pmatrix} + v_1(x) v_2(y)^\top \begin{pmatrix} 0 & 1_{\geq y}(x)\\ 1_{> y}(x) & 0 \end{pmatrix} \).
\end{align*}
Then, we have
\begin{align}
\(\jve (U-e^{\im \lambda}) \jev K_\lambda(\cdot,y)\)(x)
=\begin{cases} \begin{pmatrix}1& 0 \\ 0 & 1 \end{pmatrix} & x=y,\\ 0 & x\neq y. \end{cases}
\end{align}
\end{proposition}

\subsection{Dispersion relation}\label{subsec:constant}
We analyze $T_\lambda(x)$ starting from investigating eigenvalues. 
\begin{lemma}
Two eigenvalues of $T_\lambda(x)$ can be written as $e^{\pm \im \xi}$ for $\xi \in \C/2\pi\Z$ with $\Im \xi\geq 0$ satisfying the relation:
\begin{align}\label{disprel}
\rho(x) \cos \xi = \cos \lambda.
\end{align}
\end{lemma}

\begin{proof}
Since $\det T_\lambda(x)=1$, we can express two eigenvalues of $T_{\lambda}(x)$ as $e^{\pm \im \xi}$.
Further, without loss of generality, we can assume that $\Im \xi\geq 0$.
Finally comparing the trace of $T_{\lambda}(x)$ and the sum of $e^{\pm\im\xi}$, we obtain \eqref{disprel}.
\end{proof}

Formally, we can express $\lambda$ as $\lambda=\mathrm{Arccos}\(\rho \cos \xi\)$.
Since it is an analytic function, it can be defined on its Riemann surface $\mathcal R$.
Since $\rho(x)$ depends on $x$ in general, the Riemann surface will depend on $x$.

For $b>0$, we set 
\begin{align*}
\T_{b}:=\{z\in \C/2\pi\Z\ |\ \ 0<\Im z<b\}, \ \T_{\Re}:=\{z\in \C/2\pi\Z\ |\ \Im z=0\},
\end{align*}
and
\begin{align}
\T_{b}^j:=\{(j,z)\ |\ z\in \T_b\},\ \T_{\Re}^j=\{(j, z)\ |\ z\in \T_{\Re}\},\quad j=1, 2.
\end{align}
We set 
\begin{align}\label{def:Rb}
\mathcal R_{b}:=\T_b^1 \sqcup \T_b^2,\quad \mathcal R_{\Re}:=\T_{\Re}^1\sqcup \T_{\Re}^2,\quad \mathcal R_E:=\{(j, 0),\ (j, \pi)|\ j=1, 2\}.
\end{align}

\begin{remark}
The element of $\overline{\mathcal R_b}$ (closure of $\mathcal R_b$) should be expressed as $(\mathfrak s, z)$ where $\mathfrak s\in \{1, 2\}$.
We will sometimes use the notation $z_{\mathfrak s}:=(\mathfrak s,z)$.
Further, by abuse of notations, we express $z_{\mathfrak s}$ as $z$ with $z\in \T_b$ and comment on which sheet (i.e.\ $\mathfrak s=1$ or 2) is $z$ in if necessary.
For $\xi_j=(j,\xi)\in \overline{\mathcal R_b}$, we define $e^{\im \xi_j}:=e^{\im \xi}$ and also define trigonometric functions in $\overline{\mathcal R_b}$ by
\begin{align}\label{def:tri}
\cos \xi_{j}:=\cos \xi,\quad \sin \xi_{j}:=\sin \xi.
\end{align}
\end{remark}

Recall that the real analytic function $\left.\mathrm{Arccos}\right|_{(-1,1)}=\(\cos|_{(0,\pi)}\)^{-1}$ can be extended analytically in $\C\setminus \((-\infty,-1]\cup [1,\infty)\)$.
By the formula
\begin{align}\label{complexcos}
\cos \xi = \cos \xi_R \cosh \xi_I -\im \sin \xi_R \sinh \xi_I,\ \xi=\xi_R+\im \xi_I\in\C,
\end{align}
it is seen that $\cos(\cdot)$ maps $\{u+\im v\ | 0< u< \pi, 0<\pm v\}$ to $\mathbb{C}_{\mp}:=\{z\in\mathbb{C}| \mp \Im z>0\}$. Thus $\mathrm{Arccos}(\cdot)$ maps $\mathbb{C}_{\pm}$ to $\{u+\im v\ | 0\le u\le \pi, 0<\mp v\}$. 
\\

We recall that the condition  $0<|\alpha_{\pm}|<1$ is assumed in \eqref{ass:long}. By this, we can take sufficiently large $r>0$ such that for any $x\in\Z$ with $|x|\ge r$, we have $\delta<|\alpha(x)|<1-\delta$ for some $\delta>0$. This condition near $\pm\infty$ is important to introduce $\zeta$ in $\eqref{pointfreq}$. In a region $(-r, r)\cap\Z$, there is no problem if there exists $x_{0}\in (-r_{0}, r_{0})\cap \Z$ such that $\alpha(x_{0})=0$. However,  in what follows, we assume that $0<|\alpha(x)|<1$ for any $x\in\Z$ for simplicity.

For given $\alpha\in B_{\C}(0,1)\setminus\{0\}$, we set $b_\alpha>0$ to be the solution of $$\cosh b_\alpha =\rho^{-1}=(1-|\alpha|^2)^{-1/2}.$$ 
We notice that $b_{\alpha_{+}}=b_{\alpha_-}$ since $|\alpha_{+}|=|\alpha_{-}|$. For $\xi \in \T_{b_\alpha}$, we have $-1<\Re \(\rho \cos \xi\)<1$ by (\ref{complexcos}).
For $0<b< b_{\alpha}$, we define the function $\lambda_{\alpha}:\overline{\mathcal R_{b}}\to \C/2\pi\Z$ by
\begin{align}\label{def:lambda}
\lambda_{\alpha}(\xi):=(-1)^{j-1}\mathrm{Arccos}\(\rho\cos \xi\), \quad \xi\in \overline{\T_b^j}\quad j=1, 2.
\end{align}
For $0<b< b_{\alpha_\pm}$, we also define $\lambda_{\infty}:\overline{\mathcal R_{b}}\to \C/2\pi\Z$ by
\begin{align}\label{def:lambdainf}
\lambda_{\infty}(\xi):=\lambda_{\alpha_{\pm}}(\xi)=(-1)^{j-1}\mathrm{Arccos}\(\rho_{\infty}\cos \xi\), \quad \xi \in \overline{\T_{b}^{j}}\quad j=1, 2.
\end{align}
Since it is obvious that $\lambda_{\alpha_{+}}(\xi)=\lambda_{\alpha_{-}}(\xi)$, we unify them and write as $\lambda_{\infty}$. 
\begin{remark}
From the relation $\rho(x):=\sqrt{1-|\alpha(x)|^2}$, $\rho(x)$ is automatically determined by $\alpha(x)$. Therefore, we emphasize dependence on $\alpha$ rather than $\rho$ in \eqref{def:lambda} and \eqref{def:lambdainf}. As seen in below, we will derive several estimates by using differences of $\alpha$'s. 
\end{remark}
\begin{remark}
If $b=b_{\alpha}$, the curve of $\mathrm{Arccos}\(\rho \cos(x+\im b_{\alpha})\)$ with $-\pi < x\le \pi$ touches 0 and $\pi$. This means that $\lambda_{\alpha}$ loses injectivity. To avoid two points 0 and $\pi$, we have to take $b$ satisfying $0<b<b_{\alpha}$ (See also Figure 1). 
\end{remark}
\vspace{3mm}
\begin{center}
\input{figure1.tex}
\\
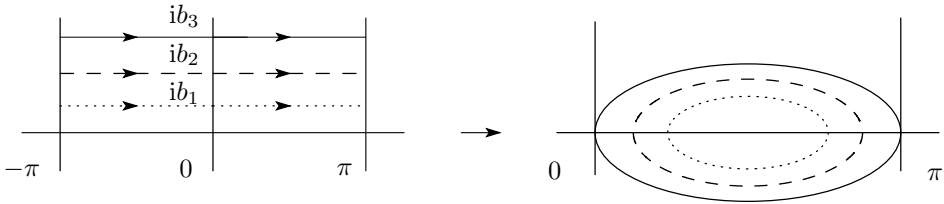
Figure 1. Ranges of $\{ x+ib_{j}| -\pi\le x\le \pi\}$ by $\mathrm{Arccos}\(\rho \cos (\cdot)\)$ with $0<b_{1}<b_{2}<b_{3}=b_{\alpha}$
\end{center}

\begin{remark}
$\mathcal R_{b_{\alpha}}$ is a subset of the Riemann surface $\mathcal R$ of $\lambda_{\alpha}(\xi)=(-1)^{j-1}\mathrm{Arccos}\(\rho \cos \xi\)$, which consists of two sheets.
If one considers $\lambda_{\alpha}$ as an analytic extension of $\lambda_{\alpha}$ initially defined on $\T_{\Re}^1$, $\lambda_{\alpha}$ will correspond to the above definition on $\mathcal R_{b_{\alpha}}$. 
\end{remark}

For $0<b< b_\alpha$, we set
\begin{align}\label{def:D}
\mathcal D_{\alpha,b}:=\lambda_{\alpha}(\mathcal R_b)=\{\lambda_{\alpha}(\xi)\ |\ \xi\in \mathcal R_b\} \subset \C/2\pi\Z,\quad \mathcal E_{\alpha}:=\lambda_{\alpha}(\mathcal R_E),
\end{align}
and
\begin{align}\label{def:barD}
\overline{\mathcal D_{\alpha,b}}:=\mathcal D_{\alpha,b}\cup\partial \mathcal D_{\alpha,b}.
\end{align}
We devide $\partial \mathcal{D}_{\alpha, b}$ as $\partial \mathcal D_{\alpha,b}:= \partial \mathcal D_{\alpha}\cup \partial \mathcal D_{\alpha,b}^{\mathrm{out}}$, where
\begin{align*}
\partial\mathcal{D}_{\alpha}&:=\cup_{\mathfrak s_1\in\{1, 2\}, \mathfrak s_2 \in \{\pm\}}\partial \mathcal D_{\alpha}^{\mathfrak s_1,\mathfrak s_2},\\
\partial \mathcal D_{\alpha}^{\mathfrak s_1,\pm}&:=\{\lambda_{\alpha}(\xi)\ |\ \xi\in \T^{\mathfrak s_1}_{\Re},\ \mp\xi\in[0,\pi]\},\quad \mathfrak s_{1}=1, 2, \\
\partial \mathcal D_{\alpha,b}^{\mathrm{out}} &:=\{\lambda_{\alpha}(\xi)\ |\ \Im \xi =b\}\cup\{\lambda_{\alpha}(\xi)|\ \xi=(0+\im\tau)_{j},\ (\pi+\im\tau)_{j},\ 0< \tau< b,\ j=1, 2\}.
\end{align*}
For $\lambda_{\infty}$, we write $\mathcal{D}_{\infty, b}$ instead of $\mathcal{D}_{\alpha_{\pm}, b}$ for instance. The other notations are also similar.

For the figure of $\partial\mathcal{D}_{\alpha}^{\mathfrak{s}_{1}, \pm}$ and $\partial\mathcal{D}_{\alpha, b}^{\mathrm{out}}$, see Figure 2, 3, 4 and 5.

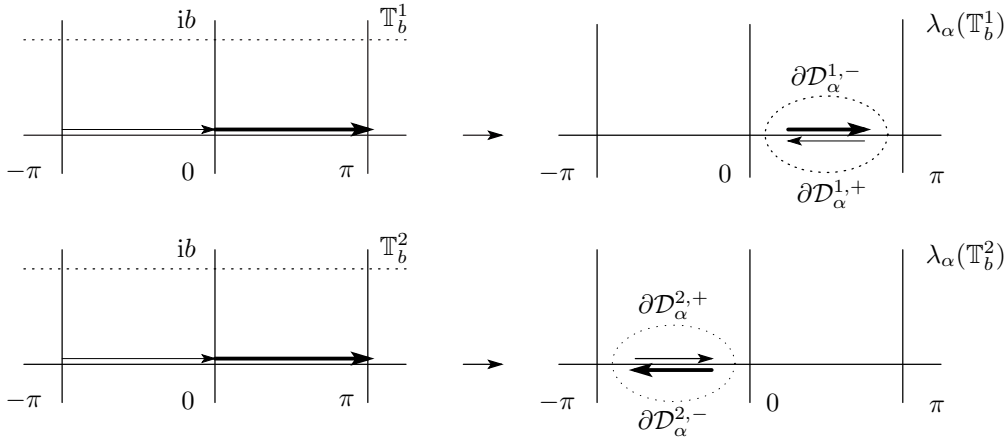
\begin{center}
\input{figure2.tex}
\end{center}
\begin{center}
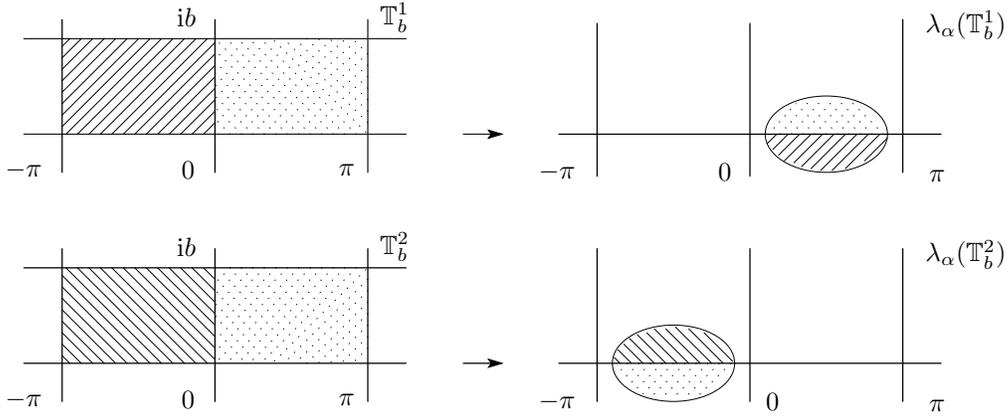
Figure 2. Domain and image of $\lambda_{\alpha}$ with $b<b_{\alpha}$. Insides of ellipses correspond to $\mathcal{D}_{\alpha, b}$. 
\end{center}

\begin{center}
\input{figure3.tex}
\end{center}
\begin{center}
Figure 3. Images of $[-\pi, 0]$ and $[0, \pi]$ by $\lambda_{\alpha}$ with $b<b_{\alpha}$. These images correspond to $\partial\mathcal{D}_{\alpha}^{\mathfrak{s}_{1}, \pm}$ ($\mathfrak{s}_{1}=1, 2$) and their union is $\partial \mathcal{D}_{\alpha}$.
\end{center}
\begin{center}
\input{figure4.tex}
\end{center}
\begin{center}
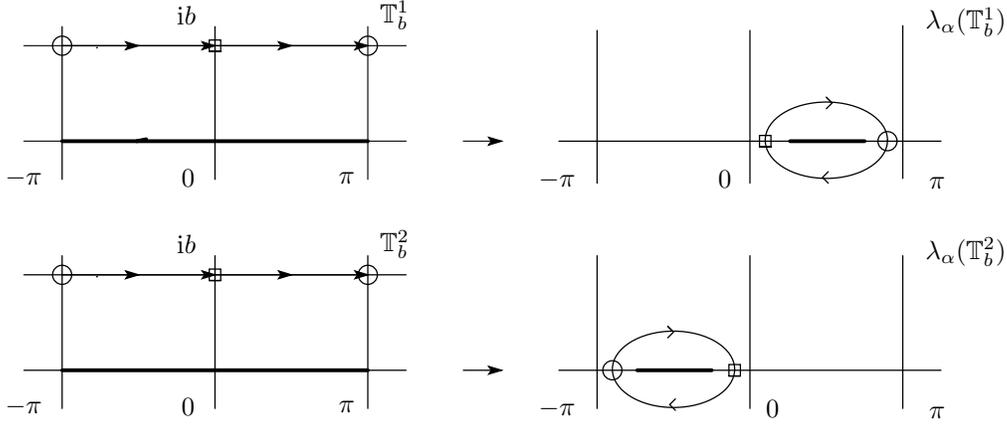
Figure 4. Curves of $\{\lambda_{\alpha}(\xi)|\ \mathrm{Im}\ \xi=b\}$ which are parts of $\partial\mathcal{D}^{\mathrm{out}}_{\alpha, b}$.
\end{center}
\begin{center}
\input{figure5.tex}
\end{center}
\begin{center}
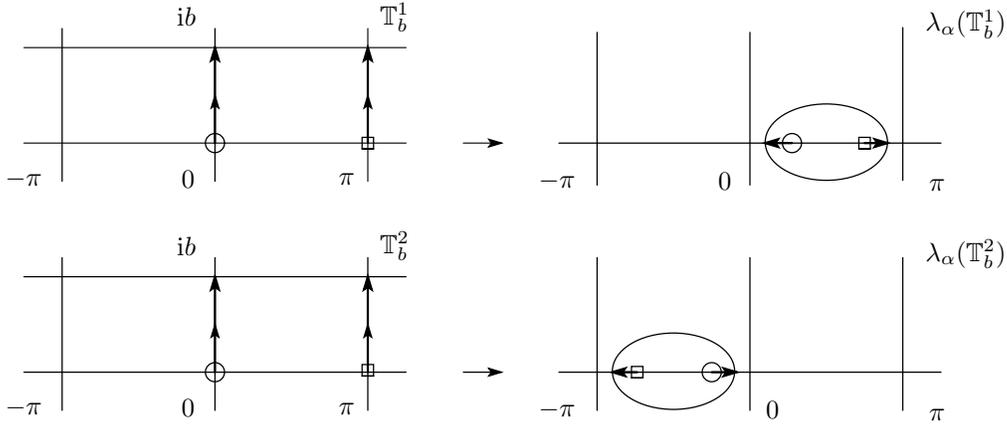
Figure 5. Images of $\{\lambda_{\alpha}(\xi)|\ \xi=(0+\im\tau)_{j},\ (\pi+\im\tau)_{j},\ 0< \tau< b,\ j=1, 2\}$ which are parts of $\partial\mathcal{D}^{\mathrm{out}}_{\alpha, b}$.
\end{center}
\begin{remark}
The domain $\mathcal D_{\alpha,b}$ consists of two connected components which are images of $\T_b^j$ ($j=1, 2$).
The above definition for $\overline{\mathcal{D}_{\alpha,b}}$ means that we distinguish $\lambda_{\alpha}(\xi)$ and $\lambda_{\alpha}(-\xi)$ for $\xi\in \mathcal R_{\Re}$ except $\xi=0_j,\pi_j$.
Thus, each connected component of $\overline{\mathcal{D}_{\alpha,b}}$ has a slit in the interior which we extend $\lambda_{\alpha}$ from above to below up to the boundary (we will not extend it through the slit).
\end{remark}
\begin{lemma}\label{lem:bihol}
Let $0<b< b_\alpha$.
Then, $\lambda_{\alpha}:\overline{\mathcal R_b}\to \overline{\mathcal D_{\alpha, b}}$ is a homeomorphism.
Moreover, $\lambda_{\alpha}:\mathcal R_b\to \mathcal D_{\alpha,b} $ is a biholomorphism.
\end{lemma}
\begin{proof}
By the definition of $\overline{\mathcal{D}_{\alpha, b}}$, it is obvious that $\left.\lambda_{\alpha}\right|_{\overline{\mathcal R_b}}$ is a surjection. Next, we show the injectivity. The injectivity of $\left.\lambda_{\alpha}\right|_{\overline{\mathbb{T}_b^j}}$ follows from the injectivity of $\cos(\cdot)$ in $\T_b^j$, $\mathrm{Arccos}(\cdot)$ in $\{z\in\C\ |\ \Re z\in (-1,1)\}$ and the definition of $\overline{\mathcal{D}_{\alpha,b}}$, where we have insert cuts to have the injectivity. If $\lambda_{\alpha}(\overline{\T_{b}^1})\cap \lambda_{\alpha}(\overline{\T_{b}^{2}})=\emptyset$ can be checked, the injectivity of $\left.\lambda_{\alpha}\right|_{\overline{\mathcal{R}}_{b}}$ follows. 
To see this, we show that for any $\xi\in\overline{\T_{b}}$, we have $0<\mathrm{Re}(\mathrm{Arccos}(\rho\cos\xi))<\pi$. First, $\mathrm{Re}(\mathrm{Arccos}(\rho\cos\xi))\in(0, \pi)$ for any $\xi\in\T_{\mathrm{Re}}$. Thus, if $\mathrm{Re}(\mathrm{Arccos}(\rho\cos\xi))=0$ or $\pi$, $\xi$ must be in $\T_{b}$. On the other hand, $\mathrm{Re}(\mathrm{Arccos}(z))=0$ or $\pi$ implies $z=\cos(\im v)$ or $z=\cos(\pi+\im v)$ for some $v\in\R$. From \eqref{complexcos}, we have $z\in\R$ and $|z|\ge 1$. However, these two conclutions are not compartible. Hence we have $0<\mathrm{Re}(\mathrm{Arccos}(\rho\cos\xi))<\pi$ for any $\xi\in\overline{\mathbb{T}_{b}}$.

Next, we show that $\lambda_{\alpha}:\mathcal R_b\to \mathcal D_{\alpha,b} $ is a biholomorphism. To see this,  it suffices to show that $\lambda_{\alpha}$ is analytic and its derivative does not vanish.
The analyiticity is  obvious and by direct computation, we have
\begin{align}\label{dxil}
\partial_\xi \lambda_{\alpha}(\xi)=(-1)^{j-1}\frac{\rho \sin \xi}{\sqrt{1-\rho^2\cos^2\xi}}\neq 0,\quad \xi\in \T_{b}^j,\quad j=1, 2.
\end{align}
Thus, $\lambda_{\alpha}$ is a biholomorphism in $\mathcal D_{\alpha,b}$.

Finally, we show that $\lambda_{\alpha}:\overline{\mathcal R_b}\to \overline{\mathcal D_{\alpha,b}}$ is a homeomorphism.
By the property $\lambda_{\alpha}(\overline{\T_{b}^{1}})\cap \lambda_{\alpha}(\overline{\T_{b}^{2}})=\emptyset$, it suffices to consider $\lambda_{\alpha}$ on restricted domains $\overline{\T_b^j}$. Here, we only consider $\overline{\T_{b}^{1}}$. The other case is similar to this. We divide $\overline{\T_{b}^{1}}$ into $\overline{\T_b^1}=\{z\in \overline{\T_b^1}\ |\ 0\leq \Re z \leq \pi\}\cup \{z\in \overline{\T_b^1}\ |\ -\pi\leq \Re z \leq 0\}$. Then we see that the function $\rho \cos \xi$ restricted on each region are continuous up to the boundaries where their images are in $\{z\in \C\ |\ \Im z\leq 0,\ -1<\Re z<1\}$ and $\{z\in \C\ |\ \Im z\geq 0,\ -1<\Re z<1\}$, respectively.
Moreover, the inverse are also continuous up to the boundaries.
Next, by considering $\cos \lambda$ defined on $\{z\in\C/2\pi\Z\ |\ 0<\Re z<\pi,\ \Im z\geq 0\}$ and 
$\{z\in\C/2\pi\Z\ |\ 0<\Re z<\pi,\ \Im z\leq 0\}$, we see that it is also continuous up to the boundaries and its inverse is also continuous.
Therefore, we have the conclusion. 
\end{proof}

We denote the inverse of $\lambda_{\alpha}$ as $\xi_{\alpha}$.
That is,
\begin{align*}
\xi_{\alpha}:\overline{\mathcal D_{\alpha,b}}  \to \overline{\mathcal R_{b}},\quad \xi_{\alpha}(\lambda):=\lambda_{\alpha}^{-1}(\lambda),
\end{align*}
for $0<b<b_{\alpha}$. We also define
\begin{align*}
\xi_{\infty}:\overline{\mathcal{D}_{\infty, b}}\to \overline{\mathcal{R}_{b}},\quad \xi_{\infty}(\lambda):=\lambda_{\infty}^{-1}(\lambda),
\end{align*} 
for $0<b<b_{\alpha_{\pm}}$. Recall that the definition of $\lambda_{\alpha}$ and $\lambda_{\infty}$ appeared in \eqref{def:lambda} and \eqref{def:lambdainf}, respectively. We notice that  $e^{\pm\im \xi_{\alpha(x)}(\lambda)}$ are eigenvalues of $T_\lambda(x)$.

In the following, we would like to introduce the function $\xi_{\alpha(x)}(\lambda_\infty(\xi))$.
The domain of $\xi_{\alpha(x)}$ depends on $x$ in general. Thus, we restrict domains of themselves to a smaller region.

We set $b_0:=\frac18 b_{\pm}$.
Then, there exists $\delta_0>0$ s.t.\ if $|\alpha-\alpha_\pm|<\delta_0$, then $2b_0<\frac12 b_\alpha<8b_0$,
$$\mathcal D_{\infty, 2b_0}\setminus \partial \mathcal D_{\alpha} \subset \mathcal D_{\alpha,\frac12b_\alpha},$$
and 
\begin{align}\label{inclution}
(\mathrm{Arccos}\rho,\mathrm{Arccos}(-\rho))\subset (\mathrm{Arccos}\(\rho_\infty \cosh b_0 \), \mathrm{Arccos}\(-\rho_\infty \cosh b_0\)).
\end{align}
\begin{remark}
We consider $\lambda_{\infty}((1, \im \tau))$ with $0<\tau<b_{0}$. By \eqref{def:lambdainf} and \eqref{complexcos}, we have $\lambda_{\infty}((1, i\tau))=\mathrm{Arccos}(\rho_{\infty}\cosh \tau)$ and its image corresponds to the left hand side of $\partial\mathcal{D}_{\infty}^{1, \pm}$ (see Figure 6). Similarily, we consider $\lambda_{\infty}((1, \pi+\im \tau))$ with $0<\tau<b_{0}$. Then, we have $\lambda_{\infty}((1, \pi+i\tau))=\mathrm{Arccos}(-\rho_{\infty}\cosh \tau)$ and its image corresponds to the right hand side of $\partial\mathcal{D}_{\infty}^{1, \pm}$. One can also consider $\lambda_{\infty}((2, \cdot))$ similarly. Inserting cuts $\{\xi\in\mathcal{R}_{2b_{0}}|\ 0\le \mathrm{Im}\ \xi\le b_{0},\ \mathrm{Re}\ \xi=0, \pi\}$ guarantees the relation \eqref{inclution} for all sufficiently large $x$.
\begin{center}
\input{figure6.tex}
\end{center}
\begin{center}
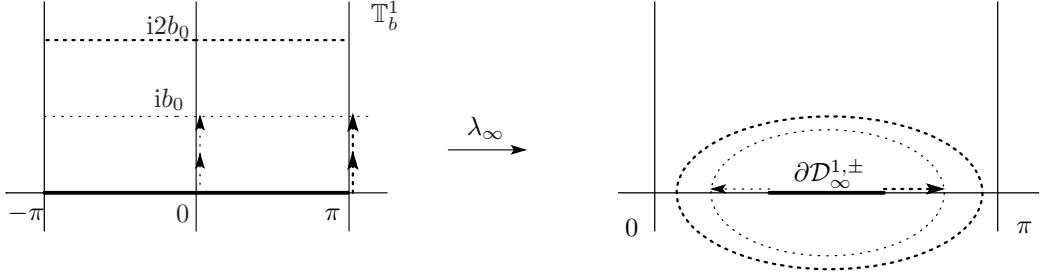
Figure 6: Images of $\{(1, \im \tau)|\ 0<\tau<b_{0} \}$ and $\{(1, \pi+\im \tau)|\ 0<\tau<b_{0}\}$ by $\lambda_{\infty}$. 
\end{center}
\end{remark}
Under above assertions, we set
\begin{align*}
\widetilde {\mathcal R_{0}}&:=\{\xi \in \mathcal R_{2b_0},\ |\ \Im \xi>b_0,\ \text{if}\ \Re \xi=0,\pi\}\cup(\mathcal{R}_{\mathrm{Re}}\setminus\mathcal{R}_{\mathrm{E}}).
\end{align*}

Then, we have
\begin{align*}
\lambda_\infty(\widetilde {\mathcal R_{0}}\setminus\mathcal{R}_{\mathrm{Re}}) \subset \mathcal D_{\alpha,\frac12 b_{\alpha}}\ \text{for}\ |\alpha-\alpha_\pm|<\delta_{0}.
\end{align*}

For the same $\delta_{0}$ introduced in the above of \eqref{inclution}, we can take sufficiently large $r_{0}>0$ such that if $\pm x\ge r_{0}$, then $|\alpha(\pm x)-\alpha_{\pm}|<\delta_{0}$. We define $\zeta(x,\xi):\Z\times \widetilde {\mathcal R_0}\to \mathcal R_{8b_0}$ by
\begin{align}\label{pointfreq}
\zeta(x,\xi):=\begin{cases}\xi_{\alpha(x)}(\lambda_{\infty}(\xi)) &  |x| \geq r_0 ,\\ \xi & |x|<r_0.\end{cases}
\end{align}

From \eqref{pointfreq}, it immediately follows that the function $\zeta(x,\xi)$ converges to $\xi$ as $|x|\to \infty$.
\begin{lemma}\label{lem:limzeta}
Let $\Omega\subset \widetilde{\mathcal{R}_{0}}$ be a compact set of $\overline{\mathcal{R}_{8b_{0}}}$. Then, $\zeta(x,\xi)\to \xi$ as $|x|\to \infty$ uniformly in $\Omega$.
\end{lemma}

\begin{proof}
Let $\Omega\subset \widetilde{\mathcal R_0}$ be a compact set of $\overline{\mathcal R_{8b_0}}$.
We claim $\xi_{\alpha}(\lambda_\infty(\xi))$ is continuous w.r.t.\ $\alpha$ and $\xi$ in $\{\alpha\in \C\ |\ |\alpha-\alpha_\pm|\leq \delta_0\}\times \Omega$.
Since this domain is compact, it is also uniformly continuous. 
Thus, for any $\varepsilon>0$, there exists $\delta>0$ s.t.\ if $|\alpha-\alpha'|<\delta$, then $\sup_{\xi \in \Omega}|\xi_{\alpha}(\lambda_\infty(\xi))-\xi_{\alpha'}(\lambda_\infty(\xi))|<\varepsilon$.
Combining this fact with $\xi_{\alpha(x)}(\lambda_\infty(\xi))\to \xi_{\infty}(\lambda_\infty(\xi))=\xi$ which also comes from the continuity of $\xi_{\alpha}$ we have the conclusion.

It remains to show the continuity of $\xi_{\alpha}(\lambda_\infty(\xi))$.
Since $\lambda_\infty$ is continuous, it suffices to show the continuity of $\xi_{\alpha}(\lambda)$ w.r.t.\ variables $\alpha$ and $\lambda$.
Let $(\alpha_n,\lambda_n)\to(\alpha_0,\lambda_0)$.
We set $\xi_n=\xi_{\alpha_n}(\lambda_n)$ and $\xi_0=\xi_{\alpha_0}(\lambda_0)$.
Then, we have $\lambda_{\alpha_n}(\xi_n)=\lambda_n\to \lambda_0$ because $\lambda_{\alpha}$ is the inverse function of $\xi_{\alpha}$.
Next, from the uniform continuity, we have $\lambda_{\alpha_n}(\xi_n)-\lambda_{\alpha_0}(\xi_n)\to 0$.
This implies $\lambda_{\alpha_0}(\xi_n)-\lambda_0=\lambda_{\alpha_0}(\xi_n)-\lambda_{\alpha_n}(\xi_n) +\lambda_{\alpha_n}(\xi_n)-\lambda_0 \to 0$.
Finally, we assume $\xi_n\not \to \xi_0$. By taking subsequence if necessary, we have $\xi_n\to \xi_1\neq \xi_0$ by the compactness of the domain.
Since $\lambda_{\alpha_0}$ is continuous we have $\lambda_{\alpha_0}(\xi_n)\to \lambda_{\alpha_0}(\xi_1)$ which is different from $\lambda_0$ because of the injectivity of $\lambda_{\alpha_0}$ proved in Lemma \ref{lem:bihol}.
Therefore, we have the conclusion.
\end{proof}
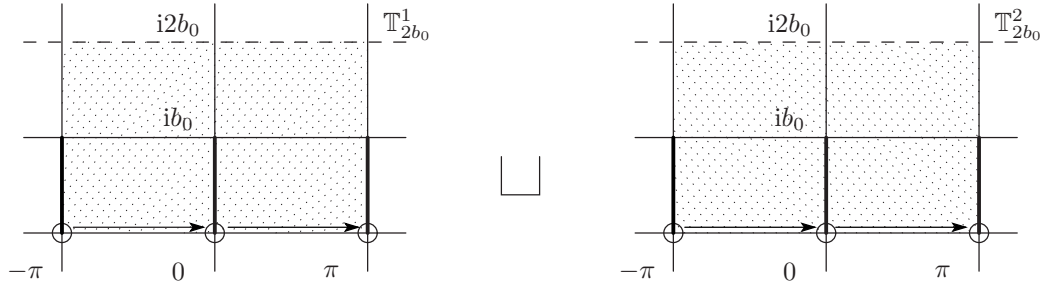
\begin{center}
\input{figure7.tex}
\end{center}
\begin{center}
Figure 7: $\widetilde{\mathcal{R}_{0}}$ consists of two sheets. For each sheets $\T_{2b_{0}}^1$ and $\T_{2b_{0}}^2$, $(-\pi, 0)$ and $(0, \pi)$ are attached and thick lines are excluded. A marks $\circ$ corresponds to an element of $\mathcal{R}_{\mathrm{E}}$.
\end{center}

The function $\zeta$ has the following symmetry.
\begin{lemma}\label{lem:symzeta}
For $\xi \in \widetilde{\mathcal R_0}$, we have $\zeta(x,-\overline{\xi})=-\overline{\zeta(x,\xi)}$. In particular, $\zeta(x, -\xi)=-\zeta(x, \xi)$ for any $\xi\in \mathcal{R}_{\mathrm{Re}}\setminus\mathcal{R}_{\mathrm{E}}$.
\end{lemma}

\begin{proof}
If $|x|<r_0$, then the statement is trivial. We only consider the case where $\xi \in \T^1_{2b_0}$ and $|x|\ge r_{0}$. The other case can be proven similarly. First, we observe that
\begin{align}\label{sym:lam}
\lambda_\infty (-\overline{\xi})=\mathrm{Arccos}(\rho_{\infty} \cos(-\overline{\xi}))=\overline{\mathrm{Arccos}(\rho_{\infty} \cos(-\xi))}=\overline{\lambda_\infty(\xi)}.
\end{align}
Since $\xi_{\alpha}$ is an inverse of $\lambda_{\alpha}$, we have
$
\lambda_{\alpha}(\xi_{\alpha}(\overline{\lambda}))=\overline{\lambda}.
$
On the other hand, it is seen that
\begin{align*}
\lambda_{\alpha}(-\overline{\xi_{\alpha}(\lambda)})=\mathrm{Arccos}\(\rho\cos(-\overline{\xi_{\alpha}(\lambda)})\) =\overline{\mathrm{Arccos}\(\rho\cos(-\xi_{\alpha}(\lambda))\)}=\overline{\lambda}.
\end{align*}
This equality implies
\begin{align}\label{sym:xi}
\xi_{\alpha}(\overline{\lambda})=-\overline{\xi_{\alpha}(\lambda)}.
\end{align}
Combining \eqref{sym:lam} and \eqref{sym:xi}, we have
$$\zeta(x, -\overline{\xi})=\xi_{\alpha(x)}(\lambda_{\infty}(-\overline{\xi}))=\xi_{\alpha(x)}(\overline{\lambda_{\infty}(\xi)})=-\overline{\xi_{\alpha(x)}(\lambda_{\infty}(\xi))}=-\overline{\zeta(x, \xi)}.$$ 
For any $\xi\in\mathcal{R}_{\mathrm{Re}}\setminus\mathcal{R}_{\mathrm{E}}$, we have $\lambda_{\infty}(\xi)\in\R$ and $\lambda_{\infty}(-\xi)=\overline{\lambda_{\infty}(\xi)}$. Thus, we have
$$\zeta(x, -\xi)=\xi_{\alpha(x)}(\lambda_{\infty}(-\xi))=\xi_{\alpha(x)}(\overline{\lambda_{\infty}(\xi)})=-\xi_{\alpha(x)}(\lambda_{\infty}(\xi))=-\zeta(x, \xi).$$
\end{proof}

For $\epsilon>0$, we define
\begin{align}\label{def:Re}
\widetilde{\mathcal R_{\epsilon}}:=\{\xi \in \widetilde {\mathcal R_0}\ |\ |\sin \xi| \geq \epsilon\}.
\end{align}

In the following, we will encounter several (large) numbers depending on $\epsilon$ all denoted by $r_\epsilon$.
Since it suffices to consider the maximum of such $r_\epsilon$'s, we will not distinguish them.

\begin{lemma}\label{lem:zetareodd}
For $\epsilon>0$, there exists $r_\epsilon>0$ s.t.\ if $|x|\geq r_\epsilon$ and $\xi \in \widetilde{\mathcal R_{\epsilon}} \cap \mathcal R_{\Re}$, then $\zeta(x,\xi) \in \mathcal R_{\Re}$ and $\zeta(x,-\xi)=-\zeta(x,\xi)$.

\end{lemma}

\begin{proof}
First, there exists $\delta=\delta_\epsilon>0$ s.t.\ $\lambda_\infty(\widetilde{\mathcal R_{\epsilon}} \cap \mathcal R_{\Re})\subset \partial \mathcal{D}_{\infty}$ and $$\lambda_{\infty}(\xi) \in (-\mathrm{Arccos}(-\rho_{\infty})+\delta,\ -\mathrm{Arccos}(\rho_{\infty})-\delta)\cup (\mathrm{Arccos}(\rho_{\infty})+\delta,\ \mathrm{Arccos}(-\rho_{\infty})-\delta),\quad \xi\in \widetilde{\mathcal{R}_{\epsilon}}\cap\mathcal{R}_{\mathrm{Re}}.$$
By taking $r_\epsilon>0$ sufficiently large, we have
\begin{align*}
\lambda_{\infty}(\xi) \in 
 (-\mathrm{Arccos}(-\rho(x)),-\mathrm{Arccos}(\rho(x)))\cup (\mathrm{Arccos}(\rho(x)),\mathrm{Arccos}(-\rho(x))),\quad \xi\in\widetilde{\mathcal{R}_{\epsilon}}\cap\mathcal{R}_{\mathrm{Re}}
\end{align*}
for all $|x|\geq r_\epsilon$.
Therefore, we have $\lambda_\infty(\xi) \in \partial \mathcal{D}_{\alpha(x)}$ and it implies $\xi_{\alpha(x)}(\lambda_\infty(\xi))\in \mathcal R_{\Re}$. The later assertion is a direct consequence of the first assertion and Lemma \ref{lem:symzeta}.
\end{proof}

In the following, we will need a bound for $|\zeta(x+1,\xi)-\zeta(x,\xi)|$.
By the definition of $\zeta$, it suffices to study the derivative of $\xi_{\alpha}$ w.r.t.\ $\alpha$.
Notice that $\xi_{\alpha}(\lambda)$ is analytic w.r.t.\ $\Re\alpha$ and $\Im \alpha$.
By the identity
\begin{align*}
\lambda_{\alpha}\(\xi_{\alpha}(\lambda)\)=\lambda,
\end{align*}
partial derivatives w.r.t.\ $X=\Re\alpha,\Im \alpha$ are
\begin{align*}
(\partial_{X}\lambda_{\alpha})(\xi_{\alpha}(\lambda))+(\partial_{\xi}\lambda_{\alpha})(\xi_{\alpha}(\lambda))\cdot\partial_{X}\xi_{\alpha}(\lambda)=0,\quad X=\Re\alpha,\Im \alpha.
\end{align*}
Thus we have
\begin{align}\label{diff:xi}
&\partial_{X} \xi_{\alpha}(\lambda)=-\frac{\partial_{X}\lambda_{\alpha}\(\xi_{\alpha}(\lambda)\)}{\partial_\xi \lambda_{\alpha}(\xi_{\alpha}(\lambda))},\quad
X=\Re \alpha, \Im \alpha.
\end{align}

In the following, we will derive several estimates. For positive quantities $A$ and $B$, we write $A\lesssim B$ if there exists a positive constant $M>0$ such that $A\le MB$. If such $M$ depends on a specific parameter $\eta$, we denote like as $A\lesssim_{\eta} B$.

By \eqref{diff:xi}, we have a quantitative version of Lemma \ref{lem:limzeta}.
\begin{proposition}\label{prop:zetaxi}
Let $\epsilon>0$.
Then, there exists $r_\epsilon>0$ s.t.\ for $\xi \in \widetilde{\mathcal R_{\epsilon}}$ and $\pm x\geq r_\epsilon$, we have
\begin{align}\label{eq:diffzetaxi}
|\zeta(\pm x,\xi)-\xi|=|\xi_{\alpha(\pm x)}(\lambda_{\infty}(\xi))-\xi_{\infty}(\lambda_{\infty}(\xi))|\lesssim \epsilon^{-1}|\alpha(\pm x)-\alpha_\pm|.
\end{align}
In particular, for $\xi\in\widetilde{\mathcal{R}_{\epsilon}}$ and $|x|\ge r_{\epsilon}$, we have $\zeta(x,\xi) \in \widetilde{\mathcal R_{\epsilon/2}}$.
\end{proposition}

\begin{proof}
First, we show that for any $\epsilon>0$, there exists $r_{\epsilon}>0$ s.t.
\begin{align}\label{diff:xi:bound2}
\sup_{|x|\geq r_\epsilon,\xi \in \widetilde{ \mathcal R_{\epsilon}}}|\partial_X \xi_{\alpha(x)}(\lambda_{\infty}(\xi) )|\lesssim \epsilon^{-1},\ X=\Re \alpha, \Im \alpha.
\end{align}
By \eqref{dxil} and \eqref{diff:xi}, it suffices to show
\begin{align*}
|\sin (\xi_{\alpha(x)}(\lambda_\infty(\xi)))|\gtrsim \epsilon,\  |x|\geq r_\epsilon,\ \xi \in \widetilde{ \mathcal R_{\epsilon}}.
\end{align*}
By Lemma \ref{lem:limzeta}, there exists $r_\epsilon$ s.t. the estimate \eqref{diff:xi:bound2} holds. Note that the implicit constant depends only on $\rho_{\infty}$. The estimate \eqref{eq:diffzetaxi} follows from \eqref{diff:xi:bound2}. The later assertion follows from \eqref{eq:diffzetaxi} and the definition of $\widetilde{\mathcal{R}_{\epsilon/2}}$ introduced in \eqref{def:Re}.
\end{proof}

\begin{proposition}\label{prop:zetal1}
Let $\epsilon>0$.
Then, there exists $r_\epsilon>0$ s.t. for $|x|\geq r_\epsilon$ and $\xi \in \widetilde {\mathcal R_\epsilon}$, we have
\begin{align}\label{eq:zeta:diff}
|\zeta(x+1,\xi)-\zeta(x,\xi)|\lesssim \epsilon^{-1}|\alpha(x+1)-\alpha(x)|.
\end{align}
In particular, for any $\xi \in \widetilde {\mathcal R_0}$, we have $\zeta(\cdot+1,\xi)-\zeta(\cdot,\xi) \in l^1(\Z)$.
\end{proposition}

\begin{proof}
As the proof of Proposition \ref{prop:zetaxi},  for $|x|\geq r_\epsilon$, we have
\begin{align*}
|\zeta(x+1,\xi)-\zeta(x,\xi)|&=|\xi_{\alpha(x+1)}(\lambda_{\infty}(\xi))-\xi_{\alpha(x)}(\lambda_{\infty}(\xi))| \lesssim \epsilon^{-1}|\alpha(x+1)-\alpha(x)|.
\end{align*}
Thus, we have \eqref{eq:zeta:diff}. $\zeta(\cdot+1,\xi)-\zeta(\cdot,\xi)\in l^1(\Z)$ follows from the above estimate and \eqref{ass:long}.
\end{proof}
We define the matrix $T(x, \xi)$ as
\begin{align}\label{trans2}
T(x, \xi):=T_{\lambda_\infty(\xi)}(x)=\rho(x)^{-1} \begin{pmatrix}e^{\im \lambda_\infty(\xi)} & \alpha(x)\\ \overline{\alpha(x)} & e^{-\im \lambda_\infty(\xi)}\end{pmatrix} ,\quad x\in\Z.
\end{align}
By the definition of $\zeta$, we see that if $|x|\geq r_{\epsilon}$ and $\xi \in \widetilde{\mathcal R_\epsilon}$,  $e^{\pm\im\zeta(x,\xi)}$ are the two eigenvalues of $T(x, \xi)$. We remark that $e^{\im \zeta(x,\xi)}\neq e^{-\im \zeta(x,\xi)}$ from Proposition \ref{prop:zetaxi}.

By direct computation, we set
\begin{align}\label{P}
P(x,\xi):=\begin{pmatrix} \alpha(x) & \alpha(x) \\ \rho(x)e^{\im \zeta(x,\xi)}-e^{\im\lambda_\infty(\xi)} & \rho(x)e^{-\im \zeta(x,\xi)}-e^{\im\lambda_\infty(\xi)} \end{pmatrix},\quad x\in\Z.
\end{align}
Then, we have
\begin{align}\label{P:det}
\det P(x,\xi)=-2\im \alpha(x)\rho(x) \sin \zeta(x,\xi).
\end{align}
Proposition \ref{prop:zetaxi} shows that the right-hand side of \eqref{P:det} is uniformly bounded from below for $|x|\geq r_\epsilon$ and $\xi \in \widetilde{\mathcal R_{\epsilon}}$.
Thus, we can diagonalize $T(x,\xi)$ such as
\begin{align*}
P(x,\xi)^{-1}T(x,\xi) P(x,\xi) = \begin{pmatrix} e^{\im \zeta(x,\xi)} & 0 \\ 0 & e^{-\im \zeta(x,\xi)}\end{pmatrix}.
\end{align*}

\begin{proposition}\label{lem:est:diffP}
We have
\begin{align}
\sum_{|x|\geq r_{\epsilon}} \| P(x+1,\xi)^{-1}P(x,\xi)-1 \|_{\mathcal L(\C^2)}\lesssim_\epsilon  \sum_{x\in\Z} |\alpha(x+1)-\alpha(x) |<\infty.
\end{align}
\end{proposition}

\begin{proof}
Since $\|P(x+1,\xi)\|_{\mathcal L(\C^2)}\lesssim_\epsilon 1$, it suffices to show
\begin{align*}
\| P(x,\xi)-P(x+1,\xi) \|_{\mathcal L(\C^2)}\lesssim_\epsilon |\alpha(x+1)-\alpha(x)|.
\end{align*}
However, this is obvious from Proposition \ref{prop:zetal1} and \eqref{P}.
\end{proof}

\subsection{Construction of modified plane waves}\label{subsec:modwave}
We are now in the position to construct modified plain waves.
We set
\begin{align*}
\nu_\pm(\xi):=\begin{pmatrix} \alpha_{+} \\ \ri e^{\pm\im\xi}-e^{\im \lambda_\infty(\xi)} \end{pmatrix},\quad \xi\in\widetilde{\mathcal{R}}_{0}
\end{align*}
and
\begin{align}\label{eq:modphase}
Z(x,\xi) := \sum_{y=0}^{x-1} \zeta(y,\xi).
\end{align}
Here we use the convention $\sum_{y=0}^{x-1}=-\sum_{y=x-1}^0 $ for $x\leq 0$. 

\begin{proposition}\label{prop:main}
Let $\xi \in \widetilde{\mathcal R}_{0}$.
Then, there exist solutions of
\begin{align}\label{modified}
\phi_\pm(x+1,\xi)=T(x,\xi)\phi_\pm(x,\xi),\quad \phi_\pm (x,\xi)-e^{\pm \im Z(x,\xi)}\nu_{\pm} (\xi)\to 0,\ x\to \pm\infty.
\end{align}
Moreover, we set $m_\pm(x,\xi):=e^{\mp\im Z(x,\xi)}\phi_\pm(x,\xi)$. Then it follows that 
\begin{align}\label{m:reg}
m_\pm \in C(\widetilde{\mathcal R_0},l^\infty(\Z,\C^2)) \cap C^{\omega}(\widetilde{\mathcal R_0},l^\infty(\Z,\C^2)).
\end{align}
\end{proposition}

\begin{remark}
If $\alpha,\theta\in l^1(\Z)$ (the short-range case), then $Z(x,\xi)-\xi x - c_\pm \to 0$ as $x\to \pm\infty$ for some constants $c_\pm$.
In the long-range case, $Z(x,\xi)$ is the desired phase correction.
\end{remark}

\begin{proof}
We only consider the $+\infty$ case.
Since $\widetilde{\mathcal R_0}=\cup_{\epsilon>0}\widetilde{\mathcal R_\epsilon}$, it suffices to show the statement of Proposition \ref{prop:main} for $\widetilde{\mathcal R_0}$ replaced by $\widetilde{\mathcal R_\epsilon}$ with arbitrary $\epsilon>0$.

For $\xi \in \widetilde{\mathcal R_\epsilon}$, we set
\begin{align}\label{def:psi}
\psi(x,\xi):=e^{-\im Z(x,\xi)}\widetilde P(x,\xi)^{-1}\phi_+(x,\xi),
\end{align}
where
\begin{align*}
\widetilde P(x,\xi)=\begin{cases} P(x,\xi) & |x|\geq r_\epsilon, \\
P_\infty(\xi):=\begin{pmatrix} \nu_+(\xi) & \nu_-(\xi)\end{pmatrix} &0\le x<r_\epsilon.
\end{cases}
\end{align*}
We remark that $P_\infty(\xi)$ diagonalizes $T_\infty(\xi):=\lim_{x\to \infty} T_{\lambda_\infty(\xi)}(x)$.
Since $Z(x+1,\xi)=Z(x,\xi)+\zeta(x,\xi)$ and
\begin{align*}
T(x,\xi)=e^{\im \zeta(x,\xi)}\widetilde P(x,\xi)\(A(x,\xi)+B(x,\xi)\){\widetilde P(x,\xi)}^{-1},
\end{align*}
where
\begin{align*}
A(x,\xi):=\begin{pmatrix} 1 & 0 \\ 0 & e^{-2\im \zeta(x,\xi)}\end{pmatrix},
\end{align*}
and
\begin{align*}
B(x,\xi):=\begin{cases} e^{-\im \zeta(x,\xi)} \widetilde P(x,\xi)^{-1}\(T(x,\xi)-T_\infty(\xi)\) \widetilde P(x,\xi) & |x|<r_\epsilon,
  \\ 0 & |x|\geq r_\epsilon. \end{cases}
\end{align*}
By substituting \eqref{def:psi} into \eqref{modified}, we obtain
\begin{align}\label{eq:psi}
\psi(x+1,\xi)=\(A(x,\xi)+V(x,\xi)\)\psi(x,\xi),\quad \psi(x,\xi)- \begin{pmatrix} 1 \\ 0 \end{pmatrix} \to 0,\ x\to \infty,
\end{align}
where
\begin{align*}
V(x,\xi)=\({\widetilde P(x+1,\xi)}^{-1}\widetilde P(x,\xi)-1\)A(x,\xi)+{\widetilde P(x+1,\xi)}^{-1}\widetilde P(x,\xi)B(x,\xi).
\end{align*}
By Proposition \ref{lem:est:diffP}, for $x\geq r_\epsilon $, we have
\begin{align*}
\sup_{\xi \in \widetilde{\mathcal R_\epsilon}}\|V(\cdot,\xi)\|_{l^1(\Z_{\cdot\geq x_{\epsilon}},\mathcal L(\C^2))}\lesssim_\epsilon \sum_{x_{\epsilon}\leq y}|\alpha(y+1)-\alpha(y)|< \infty.
\end{align*}

We divide $\Z$ into three parts.
First, we take $x_0\geq r_\epsilon $ sufficiently large so that 
$$\displaystyle\sup_{\xi \in \widetilde{\mathcal R}_\epsilon}\|V(\cdot,\xi)\|_{l^1(\Z_{\cdot\geq x_0},\mathcal L(\C^2))}\leq \displaystyle\frac{1}{2}.$$
We consider \eqref{eq:psi} in the region $x\geq x_0$.
Rewriting \eqref{eq:psi} in the Duhamel form, it suffices to solve
\begin{align}\label{eq:Neu1}
\psi(x,\xi)=\begin{pmatrix} 1 \\ 0 \end{pmatrix}- \(D(\xi)\psi(\cdot,\xi)\)(x),
\end{align}
where
\begin{align*}
\(D(\xi)v\)(x):=\sum_{y=x}^\infty \begin{pmatrix} 1 & 0 \\ 0 & e^{2\im \sum_{w=x}^y \zeta(w,\xi)} \end{pmatrix}V(y,\xi)v(y).
\end{align*}
Since $|e^{2\im \sum_{w=x}^y \zeta(w,\xi)}|\leq 1$, we have
\begin{align}
\|D(\xi)\|_{\mathcal L \(l^\infty(\Z_{\cdot\geq x_0},\C^2) \) }\leq \|V(\cdot,\xi)\|_{l^1(\Z_{\cdot\geq x_0},\mathcal L(\C^2))}\leq 1/2.
\end{align}
Therefore, we have
\begin{align*}
\psi(x,\xi)=\begin{pmatrix} 1 \\ 0 \end{pmatrix}+\sum_{n=1}^\infty \(D^n(\xi)\begin{pmatrix} 1 \\ 0 \end{pmatrix}\)(x),
\end{align*}
The above infinite series converges in $l^\infty(\Z_{\cdot\geq x_0},\C^2)$.
Therefore, we have the conclusion with $\Z$ replaced by $\Z_{\cdot\geq x_0}$.

Next, we take $x_1<x_{0}$ so that 
$$\displaystyle\sup_{\xi \in \widetilde{\mathcal R}_\epsilon}\|V(\cdot,\xi)\|_{l^1(\Z_{\cdot\leq x_1},\mathcal L(\C^2))}\leq \displaystyle\frac{1}{2}.$$
In the finite region $(x_1,x_0)\cap \Z$, we simply define $\psi$ using \eqref{eq:psi}.

Finally, to define $\psi$ in the region $x<x_1$, notice that as \eqref{eq:Neu1}, it suffices to solve
\begin{align*}
\psi(x,\xi)=\begin{pmatrix} 1 & 0 \\ 0 & e^{2\im \sum_{w=x}^{x_1-1}\zeta(w,\xi)} \end{pmatrix} \psi(x_1,\xi)-\(E(\xi)\psi(\xi,\cdot)\)(x),
\end{align*}
where
\begin{align*}
\(E(\xi)v\)(x):=\sum_{y=x}^{x_1-1} \begin{pmatrix} 1 & 0 \\ 0 & e^{2\im \sum_{w=x}^{y}\zeta(w,\xi)} \end{pmatrix}V(y,\xi)v(y). 
\end{align*}
Since we can show $\|E(\xi)\|_{\mathcal L(l^\infty(\Z_{\cdot\leq x_1},\C^2))}\leq 1/2$, we can express $\psi(x,\xi)$ as
\begin{align*}
\psi(x,\xi)=\psi(x_1,\xi)+\sum_{n=1}^\infty E(\xi)^n \psi(x_1,\xi).
\end{align*}
Therefore, we have the conclusion.
\end{proof}

\section{Proof of Theorem \ref{main:spec}}\label{sec:prmain}

We first prove the non-existence of embedded eigenvalues.

\begin{proof}[Proof of \eqref{no:emb} of Theorem \ref{main:spec}]
We only consider the eigenvalues in the upper half plane of $\C$. We suppose that there exists an eigenvalue $e^{\im\lambda}$ with $|\cos \lambda|<\rho_{\infty}$. Then, there exists $\xi \in \T\setminus\{0,\pi\}$ and  $\psi\in \mathcal H$ such that  $\lambda_\infty(\xi)=\lambda$ and $U\psi =e^{\im \lambda_\infty(\xi)}\psi$. On the other hand, $\phi_+(\cdot,\xi)$ is another solution of \eqref{eq:equiv} which is bounded on $x\geq 0$.
Corollary \ref{cor:Wron} shows that $\phi_+(\cdot,\xi)$ and $\psi$ are linearly dependent.
However, this is a contradiction.
\end{proof}

For the proof of \eqref{no:sing}, we recall some basic facts which will be a direct consequence of the limiting absorption principle.

By Proposition \ref{prop:kernel} and Proposition \ref{prop:main}, for $\xi \in \widetilde{\mathcal R_0}$, we have
\begin{align*}
\((U-e^{-\im\lambda_\infty(\xi)})^{-1}u\) (x) = \sum_{y\in\Z} K_\xi(x,y)u(y),
\end{align*}
where
\begin{align}\label{ker:explicit}
K_\xi(x,y):=&e^{-\im \lambda_\infty(\xi)} W_\xi^{-1} \(e^{\im \sum_{w=x}^{y-1}\zeta(w,\xi)}m_-(x,\xi) m_+(y,\xi)^\top \begin{pmatrix} 0 & 1_{<y}(x)\\ 1_{\leq y}(x) & 0 \end{pmatrix}\right. \\ &\quad \left.+ e^{\im \sum_{w=y}^{x-1}\zeta(w,\xi)}m_+(x,\xi) m_-(y,\xi)^\top \begin{pmatrix} 0 & 1_{\geq y}(x)\\ 1_{> y}(x) & 0 \end{pmatrix} \).\nonumber
\end{align}

Before going to next, we introduce weighted $l^2$-space. For $\sigma\in\R$, we define $l^2_{\sigma}(\Z,\C^2)$ as
\begin{align}\label{def:weightedl2}
l^2_{\sigma}(\Z,\C^2):=\left\{u:\Z\rightarrow\C^2|\ \displaystyle\sum_{x\in\Z}\|u(x)\|^{2}_{\C^2}\langle x\rangle^{2\sigma}<\infty\right\},\quad \langle u, v\rangle_{\sigma}:=\displaystyle\sum_{x\in\Z}\langle u(x), v(x)\rangle_{\C^2}\langle x\rangle^{2\sigma},
\end{align} 
where $\langle x\rangle:=\sqrt{1+x^2}$. We simply write $l^2_{\sigma}(\Z,\C^2)$ as $l^2_{\sigma}$ if there is no danger of confusion. 
\begin{proposition}\label{limiting absoption}
Let $\xi\in \mathcal R_{\Re}\setminus \mathcal R_E$.
Then, $\lim_{\varepsilon\to +0}(U-e^{\im \lambda_\infty(\xi+\im \varepsilon)})^{-1}$ exists in $\mathcal L(l^2_{\sigma},l^2_{-\sigma})$ for $\sigma>1/2$.
\end{proposition}

\begin{proof}
Since
\begin{align*}
\|(U-e^{-\im\lambda_\infty(\xi)})^{-1}\|_{\mathcal L(l^2_{\sigma},l^2_{-\sigma})}\leq\|\<x\>^{-\sigma}K_\xi(x,y) \<y\>^{-\sigma}\|_{l^2(\Z^2,\mathcal L(\C^2))},
\end{align*}
it suffices to show that $\<x\>^{-\sigma}K_{\xi+\im\epsilon}(x,y)\<y\>^{-\sigma}\to \<x\>^{-\sigma}K_\xi (x,y)\<y\>^{-\sigma}$ as $\epsilon\to +0$ in $l^2(\Z^2,\mathcal L(\C^2))$.
However, this immediately follows from Proposition \ref{prop:main} and the fact that $\Im \zeta(x,\xi)\geq 0$ for all $x\in \Z$ and $ \xi\in \widetilde{\mathcal R}_0$
\end{proof}

\begin{proof}[Proof of \eqref{no:sing} of Theorem \ref{main:spec}]
Stone's formula for unitary operators and the limiting absorption principle (Proposition \ref{limiting absoption}) imply the non-existence of the singular continuous spectrum.
\end{proof}

\appendix

\section{Gauge transformation}

In this section, we prove Proposition \ref{gauge} to simplify the form of coin operator $C_{\alpha,\beta,\theta}$ defined in \eqref{def:coin:pre}.

The existence of unitary operator $G$ has been established in \cite{O}. However, the convergence of $\alpha'$ is not considered in it. For completeness, we discuss this proposition.
\begin{proof}[Proof of Proposition \ref{gauge}]

We write $\beta(x)=e^{\im b(x)}|\beta(x)|$ with $b(x)\in \R/2\pi\Z$. We introduce two $\R$-valued functions $g, h$ on $\Z$ and set the unitary operator $G$ as
$$ (Gu)(x):=G(x)u(x),\quad G(x)=\begin{pmatrix}e^{\im g(x)} & 0 \\ 0 & e^{\im h(x)}\end{pmatrix}, \quad u\in \mathcal{H},\quad x\in\Z.$$
We will decide the values of $g(x)$ and $h(x)$, later. We have
\begin{align*}
\(GUG^{\ast}u\)(x)&=\begin{pmatrix}e^{\im g(x)} & 0 \\ 0 & e^{\im h(x)}\end{pmatrix}(SCW^{\ast}u)(x)
\\
&=\begin{pmatrix}e^{\im g(x)} & 0 \\ 0 & e^{\im h(x)}\end{pmatrix}\begin{pmatrix}(CW^{\ast}\ur)(x-1) \\ (CW^{\ast}\ul)(x+1)\end{pmatrix}
\\
&=\begin{pmatrix}e^{\im g(x)}e^{\im\theta(x-1)}e^{\im b(x-1)}e^{-\im g(x-1)}|\beta(x-1)|\ur(x-1)+e^{\im g(x)}e^{\im\theta(x-1)}e^{-\im h(x-1)}\overline{\alpha(x)}\ul(x-1)
\\
-e^{\im h(x)}e^{\im\theta(x+1)}e^{-\im g(x+1)}\alpha(x)\ur(x+1)+e^{\im h(x)}e^{\im\theta(x+1)}e^{-\im b(x+1)}|\beta(x+1)|e^{-\im h(x+1)}\ul(x+1)
\end{pmatrix}
\\
&=(S\tilde{C}u)(x),
\end{align*}
where 
$$ \tilde{C}(x)=\begin{pmatrix}e^{\im g(x+1)}e^{\im\theta(x)}e^{\im b(x)}e^{-\im g(x)}|\beta(x)| & e^{\im g(x+1)}e^{\im\theta(x)}e^{-\im h(x)}\overline{\alpha(x)}
\\
-e^{\im h(x-1)}e^{\im\theta(x)}e^{-\im g(x)}\alpha(x) & e^{\im h(x-1)}e^{\im\theta(x)}e^{-\im b(x)}e^{-\im h(x)}|\beta(x)|
\end{pmatrix}.$$

We consider following equations to vanish phases of (1, 1) and (2, 2) components of $\tilde{C}(x)$, respectively:
\begin{align}\label{1122}
g(x+1)+\theta(x)+b(x)-g(x)=0,\quad h(x-1)+\theta(x)-b(x)-h(x)=0,\quad x\in\Z.
\end{align}
From (\ref{1122}), we choose values of $g(x)$ and $h(x)$ to obtain the following equalitites.
\begin{align}\label{gh}
g(x+1)=g(x)-\theta(x)-b(x), \quad h(x-1)=h(x)-\theta(x)+b(x)\quad x\in\Z.
\end{align}
By (\ref{gh}), phases of (1, 2) and (2, 1) components of $\tilde{C}(x)$ are deformed as follows, respectively.
\begin{align}\label{12}
g(x+1)+\theta(x)-h(x)&=\{g(x)-\theta(x)-b(x)\}+\theta(x)-h(x)=-\{-g(x)-b(x)-h(x)\},
\\\label{21}
h(x-1)+\theta(x)-g(x)&=\{h(x)-\theta(x)+b(x)\}+\theta(x)-g(x)=-g(x)+b(x)+h(x).
\end{align}
From (\ref{12}) and (\ref{21}), we set $\theta'(x):=-g(x)+b(x)+h(x)$ and $\alpha'(x):=e^{\im\theta'(x)}\alpha(x)$. Then we have
$$ \tilde{C}(x)=\begin{pmatrix}|\beta(x)| & \overline{\alpha'(x)} \\ -\alpha'(x) & |\beta(x)|\end{pmatrix}.$$
Thus the former assertion holds.

Next, we consider explicit values of $g(x)$ and $h(x)$, respectively. From (\ref{gh}), we have
$$
g(x)=\begin{cases} g(0)-\displaystyle\sum^{x-1}_{k=0}\theta(k)-\displaystyle\sum^{x-1}_{k=0}b(k)\quad &(x\ge1),
\\
g(0) & (x=0),
\\
g(0)+\displaystyle\sum^{-1}_{k=x}\theta(k)+\displaystyle\sum^{-1}_{k=x}b(k) \quad &(x\le -1).
\end{cases}
$$
$$
h(x)=\begin{cases} h(0)+\displaystyle\sum^{x}_{k=1}\theta(k)-\displaystyle\sum^{x}_{k=1}b(k)\quad &(x\ge1),
\\
h(0) & (x=0),
\\
h(0)-\displaystyle\sum^{0}_{k=x+1}\theta(k)+\displaystyle\sum^{0}_{k=x+1}b(k) \quad &(x\le -1).
\end{cases}
$$
Note that $g(0)$ and $h(0)$ can be chosen arbitrarily. Finally we consider the convergence of $\alpha'$ under \eqref{ass:long:pre}. For $x\ge1$, we have
\begin{align*}
|\alpha'(x+1)-\alpha'(x)|&=|e^{\im\theta'(x+1)}\alpha(x+1)-e^{\im\theta'(x)}\alpha(x)|
\\
&\le |\theta'(x+1)-\theta'(x)||\alpha(x)|+|\alpha(x+1)-\alpha(x)|
\\
&\le 3|\theta(x+1)-\theta(x)|\cdot1+|\alpha(x+1)-\alpha(x)|\in l^1(\N).
\end{align*}
Thus the right limit $\alpha'_{+}:=\lim_{x\rightarrow\infty}\alpha'(x)$ exists. By the similar argument for $x\le -1$, we can show the existence of left limit $\alpha'_{-}:=\lim_{x\rightarrow-\infty}\alpha'(x)$. Since $|\beta(x)|\rightarrow \sqrt{1-|\alpha_{\infty}|^2}$ as $x\rightarrow\pm\infty$, we have $|\alpha'_{+}|=|\alpha'_{-}|$. Thus the result follows.
\end{proof}

\section*{Acknowledgments}
The first author was supported by the JSPS KAKENHI Grant Numbers JP19K03579, G19KK0066A, JP17H02851 and JP17H02853. The second author was supported by the JSPS KAKENHI Grant Numbers JP18K03327. The last author was supported by the JSPS KAKENHI Grant Numbers 21K13846. This work was supported by the Research Institute of Mathematical Sciences, an International Joint Usage/Research Center located in Kyoto University and by 2019 IMI Joint Use Research Program Short-term Joint Research \lq\lq Mathematics for quantum walks as quantum simulators".

Masaya Maeda
\\
Department of Mathematics and Informatics,
\\
Graduate School of Science,
\\
Chiba University,
\\
Chiba 263-8522, Japan.
\\
{\it E-mail Address}: {\tt maeda@math.s.chiba-u.ac.jp}
\vspace{3mm}
\\
\hspace{7mm}Akito Suzuki
\\
Division of Mathematics and Physics, 
\\
Faculty of Engineering, 
\\
Shinshu University, 
\\
Wakasato, Nagano 380-8553, Japan.
\\
{\it E-mail Address}: {\tt akito@shinshu-u.ac.jp}
\vspace{3mm}
\\
\hspace{7mm}Kazuyuki Wada
\\
Department of General Science and Education,
\\
National Institute of Technology, Hachinohe College, 
\\
Hachinohe 039-1192, Japan.
\\
{\it E-mail Address}: {\tt wada-g@hachinohe.kosen-ac.jp}

\end{document}